\documentclass[letterpaper, 10pt, conference]{ieeeconf}
\IEEEoverridecommandlockouts

\usepackage{amsthm}
\usepackage{amsmath, amssymb}
\usepackage{amsfonts}
\usepackage{mathtools}
\usepackage{graphicx}
\usepackage{hyperref}
\usepackage{empheq}
\usepackage{caption}
\usepackage{subcaption}

\newcommand{\norm}[1]{\left\lVert #1 \right\rVert}
\newcommand{\R}{\mathbb{R}}
\newcommand{\N}{\mathbb{N}}
\newtheorem{theorem}{Theorem}
\newtheorem{proposition}{Proposition}

\DeclareMathOperator{\diag}{diag}

\theoremstyle{definition}
\newtheorem{definition}{Definition}
\theoremstyle{remark}
\newtheorem*{remark}{Remark}

\title{\LARGE \bf Exploiting Over-The-Air Consensus for Collision Avoidance and Formation Control in Multi-Agent Systems}

\author{Michael Epp, Fabio Molinari, Jörg Raisch%
\thanks{Michael Epp, Fabio Molinari and Jörg Raisch are with the Control Systems Group, Technische Universität Berlin, Germany. Email: {\tt\small michael.epp@tu-berlin.de}, {\tt\small molinari@tu-berlin.de}}%
\thanks{Jörg Raisch is also with Science of Intelligence (SCIoI), Research Cluster of Excellence, Berlin, Germany. Email: {\tt\small raisch@control.tu-berlin.de}}
\thanks{This work was funded by the Federal Ministry of Education and Research of Germany joint project 6G-RIC, project identification number 16KISK030.}
}

\begin{document}
\maketitle

\begin{abstract}
This paper introduces a distributed control method for multi-agent robotic systems employing Over the Air Consensus (OtA-Consensus). 
Designed for agents with decoupled single-integrator dynamics, this approach aims at efficient formation achievement and collision avoidance. 
As a distinctive feature, it leverages OtA's ability to exploit interference in wireless channels, a property traditionally considered a drawback, thus enhancing communication efficiency among robots. 
An analytical proof of asymptotic convergence is established for systems with time-varying communication topologies represented by sequences of strongly connected directed graphs. 
Comparative evaluations demonstrate significant efficiency improvements over current state-of-the-art methods, especially in scenarios with a large number of agents.
\end{abstract}

\section{Introduction}
In recent years, autonomous multi-agent systems have become increasingly important in various scientific and engineering fields. The distributed control of these systems, which involves developing control algorithms and analyzing the behaviors that emerge, especially in vehicle control, has attracted much research attention \cite{paper:ren}, \cite{paper:gulzar}.

A key area of interest in these studies is \emph{formation control}, where agents need to form a specific shape or arrangement. To do this effectively, they often have to agree on a central point that represents the formation's focus. Distributed consensus protocols are used to achieve this, allowing agents to share and align information based on the limitations of their communication network \cite{paper:molinari}, \cite{paper:falconi}.

In addition to achieving formation, these systems usually have additional objectives like collision avoidance. This involves ensuring that agents always keep a safe distance from each other. Such requirements introduce complex interactions that call for not just the right consensus protocol but also a specialized control strategy. This paper focuses on using artificial potential fields for collision avoidance, a common method to prevent collisions \cite{paper:sabattini}, \cite{paper:toyota}.

Efficient energy use and communication management are often critical in these systems. Traditional formation control often neglects the communication model or depends on direct agent-to-agent communication, which can be energy-intensive in larger or densely populated networks \cite{paper:falconi}, \cite{paper:yi}, \cite{paper:almeida}. An alternative is leveraging the superposition property of wireless signals for more efficient broadcasting, a concept known as Over-the-Air (OtA) computation. This approach, identified as a promising candidate technology for 6G wireless communication systems~\cite{paper:wang}, turns traditional interference challenges into a communication advantage \cite{paper:zheng}, \cite{paper:molinari2}. 
Based on \cite{paper:molinari}, this paper proposes a control strategy that utilizes OtA broadcasting benefits while avoiding inter-agent collisions, thereby extending the work in~\cite{paper:molinari} to ensure a safe operation.

The paper is structured as follows. Section~\ref{sec:notation} summarizes notation. 
The formation control problem is defined in Section~\ref{sec:pre}. 
In Section~\ref{sec:comm_model}, the broadcast protocol is introduced. 
A control strategy is proposed in Section~\ref{sec:ctrl}. 
The system is analyzed, and its convergence properties are shown in Section~\ref{sec:convergence}.
Simulation results are shown in Section~\ref{sec:sim}, and Section~\ref{sec:conclusion} contains final remarks.
\section{Notation} \label{sec:notation}
Throughout this paper, $\R$, $\R_{>0}$, and $\R_{\geq 0}$ denote the sets of real, positive real and nonnegative real numbers, respectively. 
$\N_0$ and $\N$ will denote the nonnegative and positive integers, respectively. 
$\mathbb{I}_n$ is the identity matrix of size $n\times n$, while $1_n$ is a column vector of size $n$ with every element equal to $1$.

Given a matrix $A$, its transpose is written as $A^T$.
The entry in position $(i,j)$ of matrix $A$ is denoted $[A]_{ij}$. 
A matrix $A$ is positive, respectively nonnegative, if $\forall (i,j),\, [A]_{ij}>0$, respectively $[A]_{ij}\geq 0$. 
A square nonnegative matrix $A$ is called reducible if there exists a permutation matrix $P$ such that $PAP^T$ is of upper block triangular form. 
If $A$ is not reducible, it is \emph{irreducible}.
A nonnegative matrix $A$ is \emph{primitive} if $\exists h\in\N$ such that $A^h$ is positive.

A \emph{directed graph} $\mathcal{G}$ is a pair $(\mathcal{N}, \mathcal{A})$, where $\mathcal{N}$ is the set of nodes and $\mathcal{A}\subseteq \mathcal{N}\times\mathcal{N}$ is the set of arcs. 
$\mathcal{A}$, i.e. $(i,j)\in\mathcal{A}$ if and only if an arc goes from node $i\in\mathcal{N}$ to node $j\in\mathcal{N}$. 
The set containing all neighbors of agent $i$ is defined as 
$\mathcal{N}_i:= \left\{ j\in\mathcal{N} \,\vert\, (j,i)\in\mathcal{A}  \right\}$. A path from node $i$ to node $j$ is a sequence of arcs 
\begin{equation}
    (l_0,l_1),(l_1,l_2),\dots,(l_{p-1},l_p),
\end{equation}
with $p\in\N$, $l_0=i$ and $l_p=j$. 
The graph $\mathcal{G}$ is strongly connected if $\forall i,j\in\mathcal{N}$ there exists a path from node $i$ to node $j$.
A weighted graph is a triple $\mathcal{G} = (\mathcal{N}, \mathcal{A}, w)$, where $w:\mathcal{A}\to\R_{>0}$ assigns a positive weight to each arc.
The matrix $\mathcal{W}\in\R_{\geq 0}^{n\times n}$, with $[\mathcal{W}]_{ji}=w((i,j))$ if $(i,j)\in\mathcal{A}$ and $[\mathcal{W}]_{ji}=0$ otherwise,  is called the adjacency matrix of $\mathcal{G}$.
In this paper, we use weighted directed graphs to model the communication topology in multi-agent systems.
We allow changing topologies, meaning that the arc set and the weight function of the graph can change at discrete points in time $t_k$, $k\in\N_0$, and we write $\mathcal{G}(t_k) = (\mathcal{N}, \mathcal{A}_k, w_k)$.
\section{Preliminaries} \label{sec:pre}
Let $\mathcal{N}=\{1,\dots,n\}$, $n\in\N$, be a set of autonomous agents moving on a two dimensional plane with a decoupled single integrator dynamics, i.e.,
\begin{equation} \label{eq:sys_dynamics}
\forall i\in\mathcal{N}, \quad \dot{p}_i(t) = u_i(t),
\end{equation}
where $p_i(t)\in\R^2$ is the $i$th agents position and $u_i(t)$ its input. 
The agents exchange information at times $t_k\in\R_{\geq 0}, k\in\N_0$, such that
\begin{equation}
\exists\underline{\Delta},\bar{\Delta}\in\R_{> 0}:\quad\forall k\in\N_0,~ \underline{\Delta} \leq t_{k+1}-t_k \leq \bar{\Delta}.
\end{equation}
The control scheme will be designed to let the system reach a stationary formation. This happens when,
\begin{equation} \label{eq:cond_convergence}
    \forall i\in\mathcal{N}, \quad \lim_{t\to\infty} p_i(t) = \bar{p}^*+d_i,
\end{equation}
where $\bar{p}^*\in\R^2$ is the so-called centroid of the formation and $d_i\in\R^2$ is the desired displacement vector of agent $i$.
    
Each agent is equipped with an underlying object detection system. 
We will use this to ensure a minimum distance $\delta_s\in\R_{>0}$ between any pair of agents.
We will achieve this by a mechanism (described later) that actively steers agents away from each agent closer than $\delta_c>\delta_s$.

Let $\norm{l_{ij}(t)} := \norm{p_i(t) - p_j(t)}$ be the distance between agents $i$ and $j$. 
We assume the desired formation to be well-posed, i.e., $\forall i,j\in\mathcal{N}$,
$$ \norm{d_i-d_j} > \delta_c.$$
Let's also define the set of agents that are not in danger of colliding at time $t$ as
\begin{equation} \label{eq:cond_collision}
    \mathcal{D}_c(t) := \left\{ i\in\mathcal{N} \mid \forall j\in\mathcal{N}\setminus\{i\},~\norm{l_{ij}(t)} > \delta_c \right\}.
\end{equation}
We generalize this notion to time intervals $[t_1,t_2]$ by
\begin{equation}
    \mathcal{D}_c(t_1,t_2) := \left\{ i\in\mathcal{N} \mid\forall t\in[t_1,t_2],~i\in\mathcal{D}_c(t) \right\}.
\end{equation}
\section{Communication Model} \label{sec:comm_model}
As in~\cite{molinari2021max}, this study adopts the so-called 
\emph{Wireless Multiple Access Channel (WMAC)} to model
the superposition property (interference) of the wireless medium, i.e., 
the effect of multiple signals being simultaneously
transmitted in the same frequency band. 
Interference has traditionally been avoided by using more wireless resources, e.g., by creating orthogonal transmission (each agent is assigned its own time or frequency slot). 
\begin{definition}[Wireless Multiple Access Channel (WMAC)]
    Agents in set $\mathcal{N}$ simultaneously broadcast $\bar{\mu}_j\in\mathbb{R}$. The obtained superimposed value at the receiver is
    \begin{equation}
        \bar{\nu}_i = \sum_{j\in\mathcal{N}_i} \xi_{ij}\bar{\mu}_j,
    \end{equation}
    where $\xi_{ij}\in\mathbb{R}_{>0}$ are unknown time-varying (fading) channel coefficients; they are assumed to be a realization of a stochastic process, as in~\cite{molinari2021max}, where they are also shown to be positive
    as resulting from power modulation techniques.
\end{definition}

To deal with unknown channel coefficients,~\cite{paper:molinari2} proposed to additionally broadcast a known value, e.g., $1$ via an orthogonal channel (i.e., in a different time slot or frequency range).
In our multi-agent scenario, each agent orthogonally broadcasts the entries of a two-dimensional vector and the known value $1$.
If we use TDMA (Time Division Multiple Access), we assume that the delays between the three broadcasts are so small that they can be considered to occur at the same time $t_k$.

According to the WMAC model, all agents $i\in\mathcal{N}$ then receive 
\begin{align}
    \nu_i(t_k) &= \sum_{j\in\mathcal{N}_i} \xi_{ij}(k)\mu_j(t_k) \quad\in\R^2\\
    \nu_i'(t_k) &= \sum_{j\in\mathcal{N}_i} \xi_{ij}(k)(t_k).
\end{align}
In the following, we will assume that, $\forall i\in\mathcal{N}, \forall k\in\N_0, i\in\mathcal{N}_i$, i.e., $\xi_{ii}>0$.
Therefore $\nu_i'(t_k)$ is positive and
\begin{equation}
    \zeta_i(t_k) \coloneqq \frac{\nu_i(t_k)}{\nu_i'(t_k)}
\end{equation}
is well defined. Clearly,
\begin{equation}
    \zeta_i(t_k) = \sum_{j\in\mathcal{N}} h_{ij}(k) \mu_j(t_k),
\end{equation}
where
\begin{equation}
    h_{ij}(k) = \begin{cases}
        \frac{\xi_{ij}(k)}{\sum_{j\in\mathcal{N}_i} \xi_{ij}(k)} & \text{if } (j,i)\in\mathcal{A}(k) \\
        0 & \text{otherwise}
    \end{cases}
\end{equation}
are the normalized channel coefficients.

By construction, $h_{ij}(k)\in[0,1]$, and
\begin{equation}
    \sum_{j\in\mathcal{N}} h_{ij}(k) = 1.
\end{equation}
Hence, $\zeta_i(t_k)$ is in the convex hull of the $\mu_j(t_k)$, $j\in\mathcal{N}$.
We refer to $\zeta_i(t_k)$ as the over-the-air variable of agent $i$.
It encapsulates aggregated (superimposed) information from neighboring agents.
\section{Control Strategy} \label{sec:ctrl}
With the definitions from Section~\ref{sec:pre} and the communication model from Section~\ref{sec:comm_model} at hand, we introduce an auxiliary state variable $\vartheta_i(t)\in\R^2$, which serves as reference for position (similarly to what was done in \cite{paper:molinari}). The controlled system consists of flow and jump dynamics, accounting for (i) the continuous nature of movements in the plane and (ii) the discrete-time nature of communication.

\subsection{Flow Dynamics}
Note that, $\forall t\in[t_k,t_{k+1}], \mathcal{D}_c(t_k,t)\subseteq\mathcal{D}_c(t)\subseteq\mathcal{N}$.
To discuss the flow dynamics, we distinguish three cases:
\begin{itemize}
    \item $i\in\mathcal{D}_c(t_k,t)$, i.e., agent $i$ has not been in danger of collision since the last communication update
    \item $i\in\mathcal{D}_c(t) \setminus \mathcal{D}_c(t_k,t)$, i.e., agent $i$ is currently not in danger of collision, but was so after the last communication update
    \item $i\in\mathcal{N}\setminus\mathcal{D}_c(t)$, i.e., agent $i$ is currently in danger of collision.
\end{itemize}

Consider the following flow dynamics, $\forall t\in(t_k,t_{k+1}]$,
\begin{align}
    &\dot{p}_i(t) = u_i(t) \nonumber \\
    &= \begin{cases}
            -a(p_i(t)-\vartheta_i(t))&\text{if }i\in\mathcal{D}_c(t_k,t)
            \\
            -\cfrac{p_i(\tau_t^i)-\vartheta_i(\tau_t^i)}{t_{k+1}-\tau_t^i} & \text{if } i\in\mathcal{D}_c(t) \setminus \mathcal{D}_c(t_k,t)\\
            r_i(t)-a(p_i(t)-\vartheta_i(t))&\text{if }i\in\mathcal{N}\setminus\mathcal{D}_c(t)
        \end{cases}
    \label{eq:sys_flow_p}
    \\
    &\dot{\vartheta}_i(t) = 0,
    \label{eq:sys_flow_theta}
\end{align}
where $a\in\R_{>0}$ and, $\forall i\in\mathcal{D}_c(t)\setminus\mathcal{D}_c(t_k,t)$,
\begin{equation}
    \tau_t^i = \sup_{\substack{\tilde{t}<t \\ i\notin\mathcal{D}_c(\tilde{t})}} \tilde{t}
\end{equation}
is the most recent time where agent $i$ was in danger of colliding, and $r_i(t)\in\R^2$ is a collision avoidance term. 
We define the latter as
\begin{align}
    r_i(t)
    &= -\frac{\partial}{\partial p_i} \frac{1}{2} \sum_{i=1}^n \sum_{\substack{j=1\\j\neq i}}^n \rho_{ij}(t)
    \label{eq:ri} \\
    &= -\frac{\partial}{\partial p_i} \rho(t),
\end{align}
where 
\begin{align}
    &\rho_{ij}(t) \coloneqq 
    \begin{cases}
            \frac{\delta_c \left( \delta_c-\delta_s \right)^2}{\norm{l_{ij}(t)}-\delta_s} + \frac{\norm{l_{ij}(t)}^2}{2} &\text{if } \delta_s<\norm{l_{ij}(t)}\leq\delta_c \\
            \quad- \frac{3\delta_c^2}{2} + \delta_c\delta_s &\\
            \infty &\text{if } \norm{l_{ij}(t)}\leq\delta_s \\
            0 &\text{otherwise}
    \end{cases}
    \label{eq:rho_ij_def}
\end{align}
can be interpreted as a potential field variable.
Clearly, (\ref{eq:rho_ij_def}) is decreasing over $\norm{l_{ij}}\in(\delta_s,\delta_c]$.

\subsection{Jump Dynamics}
According to the discussion in Section~\ref{sec:comm_model}, we let agents broadcast $\forall k\in\N_0$
\begin{equation}
    \mu_i(t_k) = \begin{cases}
        p_i(t_k)-d_i
        &\text{if } i\in\mathcal{D}_c(t_k) \\
        \vartheta_i(t_k)-d_i
        &\text{otherwise}.
    \end{cases}
    \label{eq:mu_transmit}
\end{equation}
Each agent then receives 
\begin{align}
    \zeta_i(t_k) = &\sum_{j\in\mathcal{D}_c(t_k)} h_{ij}(t_k)\left(p_j(t_k)-d_j\right) \nonumber\\
    &+ \sum_{j\in\mathcal{N}\setminus\mathcal{D}_c(t_k)} h_{ij}(t_k)\left( \vartheta_j(t_k)-d_j \right). \label{eq:zeta}
\end{align}
With this, we define the jump dynamics to be
\begin{align}
    p_i(t_k^+) &= p_i(t_k)
    \label{eq:sys_jump_p}
    \\
    \vartheta_i(t_k^+) &= \zeta_i(t_k) + d_i
    \label{eq:sys_jump_theta}
\end{align}

\begin{remark}
    In what follows, the distributed dynamic system described
    by (\ref{eq:sys_flow_p})-(\ref{eq:sys_flow_theta}) and (\ref{eq:sys_jump_p})-(\ref{eq:sys_jump_theta}) is referred to as the \textit{jump-flow system}.
\end{remark}
A couple of properties are now presented, which will be of
use in the following section to prove convergence.
\begin{proposition}[Collision Avoidance]
    If $\forall i\in\mathcal{N}$, $i\in\mathcal{D}_c(0)$, then the jump-flow system
    has no collisions.
    \begin{proof}
        If $\norm{l_{ij}(t)}$ approaches the critical distance $\delta_s$, $\rho_{ij}(t)$ grows beyond all bounds. 
        The norm of the repulsive term $r_i(t)$ will then also grow beyond all bounds, effectively preventing any two agents to be within distance $\delta_s$.
    \end{proof}
\end{proposition}
The following result shows that, if agent~$i$ was in danger of colliding during the interval $(t_{k},t_{k+1})$, then its position at $t_{k+1}$ coincides with its reference.
\begin{proposition}[Converging to the reference]
    \label{prop:p_tk_theta_tk}
    For~$i\in\mathcal{N}$, $i\in\mathcal{D}_c(t_{k+1}) \setminus \mathcal{D}_c(t_k,t_{k+1})$, $p_i(t_{k+1})=\vartheta_i(t_{k+1})$.
    \begin{proof}
        $i\in\mathcal{D}_c(t_{k+1}) \setminus \mathcal{D}_c(t_k,t_{k+1})$ implies that
        $$
            t_{k}<\tau_{t}^i<t_{k+1}.
        $$
        and 
        $$
            \forall t\in(\tau^i_t,t_{k+1}], i\in\mathcal{D}_c(t)\setminus\mathcal{D}_c(t_k,t).
        $$
        By this, (\ref{eq:sys_flow_p}) and (\ref{eq:sys_flow_theta}), one can see that
        $$
            p_i(t_{k+1}) = \vartheta(\tau^i_t) = \vartheta(t_{k+1}).
        $$
    \end{proof}
\end{proposition}
\section{Convergence} \label{sec:convergence}
We will now prove that the multi-agent system described in Section~\ref{sec:ctrl} under the communication scheme from Section~\ref{sec:comm_model} will indeed converge. For this, we will employ well-known theorems from nonnegative matrix theory and Lyapunov's second method for stability.

\begin{definition} \label{def:global_min}
    A multi-agent system as defined by (\ref{eq:sys_flow_p})-(\ref{eq:sys_flow_theta}) and~(\ref{eq:sys_jump_p})-(\ref{eq:sys_jump_theta}) achieves the desired formation if
    \begin{equation}
        \forall i\in\mathcal{N},\, \lim_{t\to\infty}\, (p_i(t)-d_i) = \bar{p}^*.
    \end{equation}
    This condition implies $\forall i,\, \dot{p}_i(t)=0$ for $t\to\infty$.
\end{definition}
\begin{definition} \label{def:local_min}
    A multi-agent system as defined by (\ref{eq:sys_flow_p})-(\ref{eq:sys_flow_theta}) and~(\ref{eq:sys_jump_p})-(\ref{eq:sys_jump_theta}) is convergent if
    \begin{equation}
        \forall i,\, \lim_{t\to\infty}\, \dot{p}_i(t)=0.
    \end{equation}
    Clearly, convergence is a necessary condition for achieving the desired formation.
\end{definition}

To simplify the following analysis of the system and enhance readability, we refer to the orthogonal coordinates of the plane as $x$ and $y$, and we denote the $x$-coordinate of a vector by a superscript $x$, e.g. the $x$-coordinate of vector $p_i(t)$ is $p_i^x(t)$. 
Additionally, we write
\begin{equation}
    p^x(t) = \left[ p_1^x(t),\dots,p_n^x \right]^T
\end{equation}
and
\begin{equation}
    p(t) = [p_1^T(t),\dots,p_n^T(t)]^T.
\end{equation}

\begin{theorem} \label{thm:convergence}
    Given a series of strongly connected graphs $\mathcal{G}(t_k)$, $k\in\N_0$, and a well-posed formation control problem, the system defined by (\ref{eq:sys_flow_p})-(\ref{eq:sys_flow_theta}) and~(\ref{eq:sys_jump_p})-(\ref{eq:sys_jump_theta}) and initial conditions satisfying $\forall i,\, i\in\mathcal{D}_c(0)$ is convergent.
\end{theorem}
\begin{proof}
    In order to prove Theorem~\ref{thm:convergence}, we first transform coordinates:
    \begin{align}
        \tilde{p}_i(t) &:= p_i(t) - d_i \\
        \tilde{\vartheta}_i(t) &:= \vartheta_i(t)-d_i.
    \end{align}
    From this it is clear that $\tilde{p}(t) = p(t)-d$, $\tilde{p}_i^x(t) = p_i^x(t) - d_i^x$ and $\tilde{p}^x(t) = p^x(t)-d^x$, where $d=[d_1^T,\dots,d_n^T]^T$ and $d^x=[d_1^x,\dots,d_n^x]^T$, while the variables $\tilde{\vartheta}(t)$, $\tilde{\vartheta}_i^x(t)$ and $\tilde{\vartheta}^x(t)$ are defined analogously.

    With these definitions at hand, we can rewrite the system dynamics of the jump-flow system as
    \begin{align}
    &\dot{\tilde{p}}_i(t) \nonumber \\
    &= \begin{cases}
            -a(\tilde{p}_i(t)-\tilde{\vartheta}_i(t))&\text{if }i\in\mathcal{D}_c(t_k,t)
            \\
            -\cfrac{\tilde{p}_i(\tau_t^i)-\tilde{\vartheta}_i(\tau_t^i)}{t_{k+1}-\tau_t^i} & \text{if } i\in\mathcal{D}_c(t) \setminus \mathcal{D}_c(t_k,t)\\
            r_i(t)-a(\tilde{p}_i(t)-\tilde{\vartheta}_i(t))&\text{if }i\in\mathcal{N}\setminus\mathcal{D}_c(t)
        \end{cases}
    \label{eq:sys_flow_p_tilde}
    \\
    &\dot{\tilde{\vartheta}}_i(t) = 0
    \label{eq:sys_flow_theta_tilde}
    \end{align}
    between update times, i.e. $\forall t\in(t_k^+,t_{k+1}]$, and
    \begin{align}
        \tilde{p}_i(t_k^+) &= \tilde{p}_i(t_k) \label{eq:sys_jump_p_tilde}\\
        \tilde{\vartheta}_i(t_k^+) &= \zeta_i(t_k) \label{eq:sys_jump_theta_tilde}
    \end{align}
    at update times $t_k, k\in\N_0$.
    Furthermore, we can rewrite (\ref{eq:zeta}) in terms of the transformed variables as
    \begin{equation} \label{eq:zeta_tilde}
        \zeta_i(t_k) = \sum_{j\in\mathcal{D}_c(t_k)} h_{ij}(t_k)\tilde p_j(t_k) + \sum_{j\in\mathcal{N}\setminus\mathcal{D}_c(t_k)} h_{ij}(t_k)\tilde\vartheta_j(t_k).
    \end{equation}
    For further analysis, we require the following proposition.

    \begin{proposition} \label{prop:transmit_convex}
        At every update time $t_k$, $k\in\N_0$, each agent transmits a convex combination of $\tilde{p}_i(t_{k-1})$ and $\tilde\vartheta_i(t_k)$.
    \end{proposition}
    \begin{proof}
        To proof this proposition, we will show that $\mu_i(t_k)$ can be expressed as
        \begin{equation} \label{eq:mu_transmit_convex}
            \mu_i(t_k) = \lambda_i(t_k) \tilde{p}_i(t_{k-1}) + (1-\lambda_i(t_k))\tilde{\vartheta}_i(t_k),
        \end{equation}
        where $\lambda_i(t_k)\in [0,1)$.
        For $i\in\mathcal{N}\setminus\mathcal{D}_c(t_k)$, according to (\ref{eq:mu_transmit}), $\mu_i(t_k)=\tilde{\vartheta}_i(t_k)$, hence (\ref{eq:mu_transmit_convex}) holds for $\lambda_i(t_k)=0$. 

        For $i\in\mathcal{D}_c(t_k)\setminus\mathcal{D}_c(t_{k-1},t_k)$. Proposition~\ref{prop:p_tk_theta_tk} implies $p_i(t_k) = \vartheta_i(t_k)$, hence (\ref{eq:mu_transmit}) implies $\mu_i(t_k)=\tilde{\vartheta}_i(t_k)$ as well.
        Therefore, (\ref{eq:mu_transmit_convex}) holds for $\lambda_i(t_k)=0$.

        Finally, for $i\in\mathcal{D}_c(t_{k-1}^+,t_k)$, agent dynamics in the interval $(t_{k-1}^+,t_k]$ is governed by
        \begin{equation}
            \dot{\tilde p}_i(t) = -a(\tilde p_i(t)-\tilde\vartheta_i(t)).
        \end{equation}
        Consequently, $\forall t\in(t_{k-1}^+,t_k]$,
        \begin{equation}
            \tilde{p}_i(t) = e^{-a(t-t_{k-1})}\tilde{p}_i(t_{k-1}) + (1-e^{-a(t-t_{k-1})})\tilde\vartheta_i(t_k)
        \end{equation}
        and
        \begin{equation}
            \tilde{p}_i(t_k) = e^{-a(t_k-t_{k-1})}\tilde{p}_i(t_{k-1}) + (1-e^{-a(t_k-t_{k-1})})\tilde\vartheta_i(t_k).
        \end{equation}
        As, for $t>t_{k-1}$, $e^{-a(t-t_{k-1})} \in (0,1)$, (\ref{eq:mu_transmit_convex}) holds for $\lambda_i(t_k)=e^{-a(t_k-t_{k-1})}$. The proposition is therefore proven.
    \end{proof}
    Given Proposition~\ref{prop:transmit_convex}, (\ref{eq:zeta_tilde}) can be reformulated as
    \begin{align}
        \zeta_i(t_k) &= \nonumber\\ 
        &\sum_{j\in\mathcal{N}} h_{ij}(t_k) 
        \left( \lambda_j(t_k) \tilde{p}_j(t_{k-1}) + (1-\lambda_j(t_k))\tilde{\vartheta}_j(t_k) \right),
        \label{eq:zeta_convex}
    \end{align}
    with $\forall i,\,\lambda_i(t_k)\in [0,1)$.

    In the following we will show that for $t\to\infty$ consensus is achieved in $\tilde\vartheta_i(t)$, i.e. $\forall (i,j),\, \tilde\vartheta_i(t) = \tilde\vartheta_j(t)$.
    For simplicity reasons we will only consider the $x$-coordinates of $\tilde\vartheta_i(t)$ here, however, the analysis is identical for the $y$-coordinate.

    From (\ref{eq:sys_flow_theta_tilde}), (\ref{eq:sys_jump_theta_tilde}) and~(\ref{eq:zeta_convex}) we can formulate the update law for $\tilde\vartheta(t_k^+)$ as
    \begin{align}
        \tilde\vartheta^x(t_{k+1}) &= \tilde\vartheta^x(t_k^+) = \zeta^x(t_k) \nonumber\\
        &= H_k \mu^x(t_k) \nonumber\\
        &= H_k\left[ \Lambda_k\tilde{p}^x(t_{k-1}) + \left( \mathbb{I}_n-\Lambda_k \right) \tilde\vartheta^x(t_k) \right],
        \label{eq:theta_tilde_x}
    \end{align}
    where $[H_k]_{ij}=h_{ij}(t_k)$ and $\Lambda_k=\diag(\lambda_1(t_k),\dots,\lambda_n(t_k))$.

    \begin{proposition} \label{prop:p_mu}
        For all $k\in\N_0$, it holds that $\lambda_i(t_k)\tilde{p}_i(t_{k-1}) = \lambda_i(t_k)\mu_i(t_{k-1})$.
    \end{proposition}
    \begin{proof}
        From the proof of Proposition~\ref{prop:transmit_convex}, if $i\in\mathcal{N}\setminus\mathcal{D}_c(t_{k})$, $\lambda_i(t_k)=0$, immediately establishing the equality.

        If $i\in\mathcal{D}_c(t_k)\setminus\mathcal{D}_c(t_{k-1}^+,t_k)$, the proof of Proposition~\ref{prop:transmit_convex} again, implies that $\lambda_i(t_k)=0$, hence the equality in Proposition~\ref{prop:p_mu} is trivially satisfied.
        
        Finally, from continuity of positions at update times it follows that $i\in\mathcal{D}_c(t_{k-1}^+,t_k) \implies i\in\mathcal{D}_c(t_{k-1})$ and therefore $\mu_i(t_{k-1})=\tilde{p}_i(t_{k-1})$.
        For $i\in\mathcal{D}_c(t_{k-1}^+,t_k)$, the equality in Proposition~\ref{prop:p_mu} is therefore satisfied for any $\lambda_i(t_k)$.
    \end{proof}

    With Proposition~\ref{prop:p_mu} and the fact that $\tilde\vartheta^x(t_k) = H_{k-1}\mu^x(t_{k-1})$, (\ref{eq:mu_transmit_convex}) can be rewritten as
    \begin{equation}
        \mu^x(t_{k}) = \left[ \Lambda_k + \left( \mathbb{I}_n-\Lambda_k \right) H_{k-1} \right] \mu^x(t_{k-1}).
    \end{equation}
    Therefore, 
    \begin{equation}
        \tilde\vartheta_i(t_{k+1}) = H_k \prod_{j=1}^k \left[ \Lambda_j + \left( \mathbb{I}_n-\Lambda_j \right) H_{j-1} \right] \mu^x(0),
        \label{eq:theta_tilde_x_2}
    \end{equation}
    with $\mu^x(0) = \tilde{p}^x(0) = \tilde\vartheta^x(0)$ by definition.

    \begin{proposition} \label{prop:primitivity}
        Given of strongly connected graph $\mathcal{G}(t_k)$, the matrix $\Lambda + (\mathbb{I}_n-\Lambda)H_k$ is row-stochastic, irreducible and primitive for any diagonal matrix $\Lambda=\diag(\lambda_1,\dots,\lambda_n)$ with $\forall i,\,\lambda_{i}\in[0,1)$.
    \end{proposition}
    \begin{proof}
        As shown in Section~\ref{sec:comm_model}, $\forall i,\, \sum_{j\in\mathcal{N}} h_{ij}(t_k) = 1$, hence, the matrix $H_k$ is row-stochastic. Since $\forall i,\,\lambda_i\in[0,1)$, by matrix multiplication it is trivial to show that $\Lambda+(\mathbb{I}_n-\Lambda)H_k$ is row-stochastic as well.

        Note that the matrix $H_k$ is the weighted adjacency matrix to the graph $\mathcal{G}(t_k)$. Since $\mathcal{G}(t_k)$ is strongly connected, by~\cite[Theorem 6.2.24]{book:horn}, $H_k$ is irreducible. 
        As $\forall i,\, \lambda_i \in [0,1)$, $(\mathbb{I}_n-\Lambda)H_k$ is irreducible. 
        By~\cite[Theorem 1]{paper:schwarz}, the sum of an irreducible matrix and a nonnegative matrix is irreducible. 
        Hence, the matrix $\Lambda + (\mathbb{I}_n-\Lambda)H_k$ is irreducible.

        Lastly, by~\cite[Chapter 3, Corollary 1.1]{book:minc}, any irreducible matrix with positive trace is primitive. Since $\forall i,\, h_{ii}(t_k)>0$ by construction, $H_k$ and therefore $\Lambda + (\mathbb{I}_n-\Lambda)H_k$ have positive diagonal entries and are therefore primitive.
    \end{proof}

    \begin{proposition} \label{prop:consensus}
        Given a series of strongly connected graphs $\mathcal{G}(t_k)$ at update times $t_k$, $k\in\N_0$, the system achieves consensus in $\tilde\vartheta^x(t)$, i.e.
        \begin{equation}
            \lim_{t\to\infty} \tilde\vartheta^x(t) = 1_n\hat{\vartheta}^x,
        \end{equation}
        where $\hat{\vartheta}^x \in\R$.
    \end{proposition}
    \begin{proof}
        Since in between update times $\dot{\tilde\vartheta}=0$, the requirement can be reformulated as
        \begin{equation}
            \lim_{k\to\infty} \tilde\vartheta^x(t_k) = 1_n\hat{\vartheta}^x.
        \end{equation}
        By (\ref{eq:theta_tilde_x_2}),
        \begin{equation}
            \lim_{k\to\infty} \tilde\vartheta^x(t_k) = \lim_{k\to\infty} H_k \prod_{j=1}^k \left[ \Lambda_j + \left( \mathbb{I}_n-\Lambda_j \right) H_{j-1} \right] \tilde{p}^x(0).
        \end{equation}
        By~\cite{paper:ren},~\cite{paper:wolfowitz}, an infinite product of primitive row-stochastic square matrices of dimension $n$ converges to
        \begin{equation}
            \lim_{k\to\infty} H_k \prod_{j=1}^k \left[ \Lambda_j + \left( \mathbb{I}_n-\Lambda_j \right) H_{j-1} \right] = 1_n v^T,
        \end{equation}
        with $v\in\R^n_{>0}$ and $1_n^T v = 1$. 
        Hence, $\tilde\vartheta^x(t)$ converges to
        \begin{equation}
            \lim_{t\to\infty} \tilde\vartheta^x(t) = 1_nv^T p^x(0) = \hat{\vartheta}^x1_n,
        \end{equation}
        with $\hat{\vartheta}^x = v^T p^x(0) \coloneqq \bar{p}^x$.
    \end{proof}

    Since the same analysis and especially Proposition~\ref{prop:consensus} holds for the $y$-coordinates as well, we can generalize the result to
    \begin{equation}
        \forall i,\, \lim_{t\to\infty} \tilde\vartheta_i(t) = \bar{p}^* = (\bar{p}^x,\bar{p}^y)^T.
    \end{equation}
    This shows that $\forall i,\, \tilde\vartheta_i(t)$ converge to a common point $\bar{p}^*\in\R^2$ and the system agrees on a centroid. 
    In order to show that ${p}(t)$ converges, we consider the system for $t\to\infty$, i.e., $\forall i,\, \tilde\vartheta_i(t)=\bar{p}^*$, and define the Lyapunov function $V(t)\in\R_{\geq 0}$
    \begin{equation} \label{eq:lyapunov}
        V(t) = \frac{1}{2}a \sum_{i=1}^n ({p}_i(t)-{p}_i^*)^T({p}_i(t)-{p}_i^*) + \rho(p(t)),
    \end{equation}
    where $p_i^*=\bar{p}^*+d_i$. 
    Taking the time derivative of $V(t)$ and considering (\ref{eq:sys_flow_p}) leads to
    \begin{align}
        \dot{V}(t) &= a \sum_{i=1}^n \dot{{p}}_i^T(t)({p}_i(t)-{p}_i^*) - \dot{{p}}_i^T(t) r_i(t) \nonumber\\
        &= - \sum_{i=1}^n \dot{{p}}_i^T(t)\left[ -a({p}_i(t)-{p}_i^*) + r_i(t) \right], \nonumber\\
        &= \sum_{i=1}^n \begin{cases}
            -\dot{{p}}^T_i(t)\dot{{p}}_i(t) 
            & \text{if $\ast$} \\
            -\frac{a\alpha_i(t)( p_i(\tau_t^i)-{p}_i^*)^T(p_i(\tau_t^i)-{p}_i^*)}{t_{k+1}-\tau_t^i}
            &\text{if $\star$}.
        \end{cases}
        \label{eq:lyapunov_dot}
    \end{align}
    For notational purposes, we let $\ast$ denote $i\in\mathcal{D}_c(t_k,t)\cup\mathcal{N}\setminus\mathcal{D}_c(t)$, and $\star$ denotes $i\in\mathcal{D}_c(t)\setminus\mathcal{D}_c(t_k,t)$, where $t_k$ is the last update time before $t$.
    
    In the first line we used the fact that, by (\ref{eq:ri}),
    \begin{align}
        \frac{\text{d}}{\text{d}t} \rho(p(t)) &= \dot{p}^T(t)\frac{\partial\rho}{\partial p}(t) \nonumber\\
        &=-\dot{p}^T(t) r(t) = -\sum_{i=1}^n \dot{p}_i^T(t)r_i(t).
    \end{align}
    In the last equality of (\ref{eq:lyapunov_dot}) we expressed ${p}_i(t)$ for the case $\star$ as ${p}_i(t) = \alpha_i(t) {p}_i(\tau_t^i) + (1-\alpha_i(t)){p}_i^*$, with
    \begin{equation}
        \alpha_i(t) = 1-\frac{t-\tau_t^i}{t_{k+1}-\tau_t^i}.
    \end{equation}
    This follows immediately from (\ref{eq:sys_flow_p}) if $\vartheta_i = p_i^*$.
    Clearly, $\alpha_i(t)\in [0,1]$ for $t\in[\tau_t^i,t_{k+1}]$, and $\alpha_i(t) = 0$ if and only if $t = t_{k+1}$, i.e, $p_i(t)=p_i^*$ (Proposition~\ref{prop:p_tk_theta_tk}), implying $\dot{p}_i(t)=0$.

    Hence, all summands in (\ref{eq:lyapunov_dot}) are nonpositive and $\dot{V}(t)\leq 0$. Moreover,
    \begin{equation}
        \dot{V}(t)=0 \iff \forall i,\, \dot{p}_i(t)=0,
        \label{eq:lyapunov_dot_iff_p_dot}
    \end{equation}
    which, aside from all agents being in the desired formation, can occur if the repulsive forces exactly oppose the agents' desired movement.
    
    Therefore, it is evident that $V(t)$ decreases to a certain value $\gamma\in\R_{\geq 0}$, that is,
    \begin{equation}
        \lim_{t\to\infty} V(t) = \gamma.
    \end{equation}
    In the case that $\gamma=0$, by (\ref{eq:lyapunov}) and Definition~\ref{def:global_min}, the system reaches its targeted formation and $V(t)$ achieves a global minimum. 
    If $\gamma>0$, the function $V(t)$ is in a local minimum, and the multi-agent system has converged to positions that are not consistent with the targeted formation, i.e., 
    \begin{equation}
        \lim_{t\to\infty} p_i(t) \neq p_i^*
    \end{equation}
    for some agents $i$.

    In conclusion, we have shown Theorem~\ref{thm:convergence} (convergence of the multi-agent system (\ref{eq:sys_flow_p}), (\ref{eq:sys_flow_theta}), (\ref{eq:sys_jump_p}), (\ref{eq:sys_jump_theta}), but not necessarily to the targeted formation).

    Simulation experiments in the following section indicate that convergence to a local minimum, i.e., the system ``getting stuck'' in a steady state that does not correspond to the desired formation, may be caused by perfect symmetries in the setup. 
    This would indicate that such an outcome corresponds to a non-generic scenario. 
\end{proof}
\section{Simulation Results} \label{sec:sim}
We consider a set of $n=6$ agents with randomly chosen initial positions such that $\forall i,\, i\in\mathcal{D}_c(0)$. For all agents, we choose $a=1$. The displacement vectors $d_i$ are chosen such that the desired formation is a regular hexagon. For each agent, the safety and critical radii are $\delta_s=4$ and $\delta_c=8$. We assume a constant update interval of $\forall k\in\N_0,\, t_{k+1}-t_k=0.1$s.

At update times $t_k$, $k\in\N_0$, a network topology, i.e. $\mathcal{A}(t_k)$, is randomly chosen out of a set of five different strongly connected topologies. Furthermore, the channel fading coefficients are randomly chosen from a uniform distribution, i.e. $\forall (i,j)\in\mathcal{A}(t_k),\, \xi_{ij}(t_k) \sim \mathcal{U}(0,1)$.

The simulation is carried out using the Runge-Kutta integration method of fourth order with a step size of $10^{-3}$s and a total simulation time of $20$s. 
Figure~\ref{fig:agents_traj} shows the trajectories of the described system, where the initial and end positions are marked with circles and crosses, respectively. Clearly, the multi-agent system converges to the targeted hexagonal formation. 
The minimal distance between two agents at any time amounts to $6.18$, i.e., collisions are avoided.

\begin{figure}
    \centering
    \includegraphics[width=\linewidth]{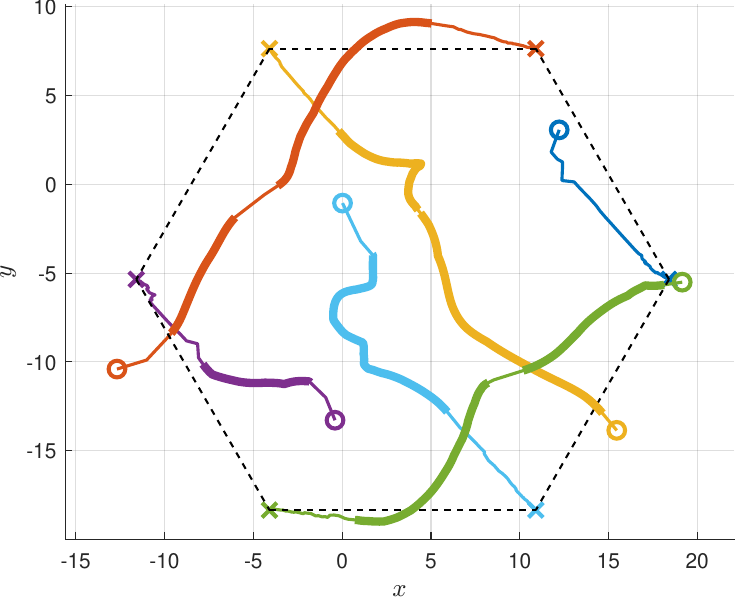}
    \caption{Trajectories of six agents in space seeking a given hexagonal-shaped formation while avoiding collisions. Circles and crosses denote initial and end positions, respectively. Wider line sections indicate when agents were in danger of colliding.}
    \label{fig:agents_traj}
\end{figure}

To highlight the efficiency of the employed broadcast communication protocol, we count the number of orthogonal transmissions until the agents agree on a common centroid and compare that to the number of transmissions if a standard orthogonal channel access protocol was to be employed.
To this end, we say the system has agreed on a centroid when for all $t>t_k$ the variance of all $\tilde\vartheta_i(t_k)$ is below a threshold of $0.01$.

In this experiment it required $71$ communications steps until the system agreed on a centroid. Utilizing the suggested OtA broadcast protocol leads to a total transmission number of $213$, since three values need to be transmitted at every communication update. 
On the other hand, if we were to employ a node-to-node communication protocol, it would take the system only $68$ communication steps to agree on a centroid. The reason for this slightly faster convergence is that with this communication protocol, each agent $i$ has knowledge of the exact transmitted values of its neighboring agents at every update time $t_k$, thus enabling choosing $\tilde\vartheta_i(t_k)$ as the arithmetic mean of all $\mu_j(t_k)$, $j\in\mathcal{N}_i$.
However, counting the number of individual transmissions leads to a total of $2214$ transmissions, showing that while the agents agree slightly faster on a common centroid, the number of transmissions increases tenfold.

Another conceivable approach is to employ a broadcast communication protocol which does not exploit interference, but rather uses time- or frequency-division multiplexing to exchange information among agents. 
In this case, at each update time, each agent would be assigned two orthogonal channels for a total of $2n$ orthogonal channels (compared to the three orthogonal channels required if interference is exploited). 
Similar to the case of node-to-node communication, this would allow for a slightly faster agreement of the agents, while keeping the number of individual transmissions low, and in this case lower when compared to the employed protocol in this paper, with $136$ individual transmissions.
However, for systems with a large number of agents, the superior scalability of the OtA broadcast approach, employed in this paper, will result in a drastically smaller number of required orthogonal channels.
This shows that the described broadcast protocol requires considerably less resources than both standard approaches, and this advantage can be expected to grow with the number of agents.

As pointed out in Section~\ref{sec:convergence}, the proposed control algorithm, does not ensure convergence to the desired formation. 
Instead, it may converge to a local minimum, in which the agents reach final positions which are not consistent with the targeted formation. 
This case is illustrated in Figure~\ref{fig:agents_traj_lm} for a numerical experiment with a system of four agents. 
We use the same parameters as in the experiment above, but with symmetric initial positions and a constant fully connected and balanced network topology, i.e., $\forall k\in\N_0,\,\forall (i,j),\, h_{ij}(t_k)=0.25$. 
In addition to the initial and final positions, the targeted positions of each agent are marked by diamonds of the respective color.
Clearly, due to the perfect symmetry of the numerical experimental setup and chosen network topology, the collision avoidance forces prevent the system from converging to the targeted formation.

However, employing a random sequence of strongly connected network topologies, as in the first experiment, introduces sufficient asymmetries into the system, allowing it to converge to the targeted formation.
Figure~\ref{fig:agents_traj_lm_solved} shows the agents reaching the desired formation in this case.
This shows that while the proposed control algorithm in general does not ensure convergence to the desired formation, small asymmetries in the experimental setup may suffice to avoid local minima and to achieve the targeted formation.
This would be expected in practical applications, where such imperfections are an inherent part of the system through, e.g., time-varying channel coefficients.

\begin{figure}
    \centering
    \begin{subfigure}{\linewidth}
        \includegraphics[width=\linewidth]{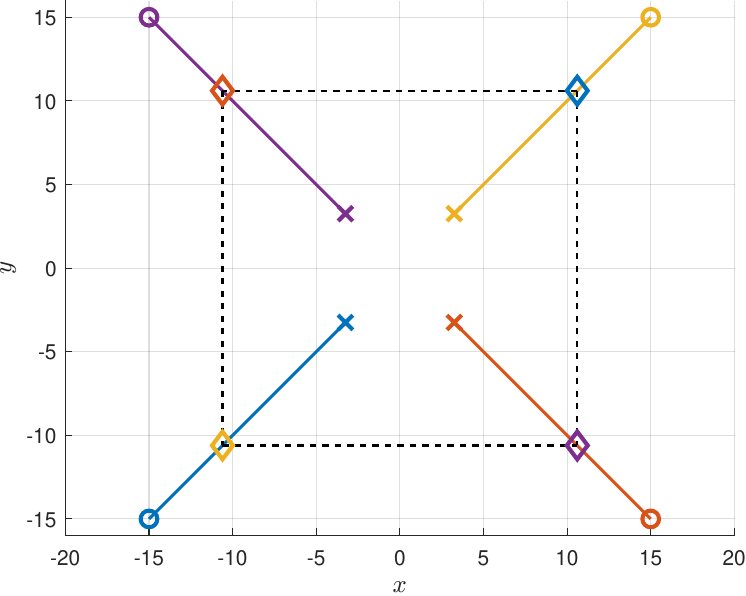}
        \caption{Agents converge to a local minimum.}
        \label{fig:agents_traj_lm}
        \vspace*{1em}
    \end{subfigure}
    \begin{subfigure}{\linewidth}
        \includegraphics[width=\linewidth]{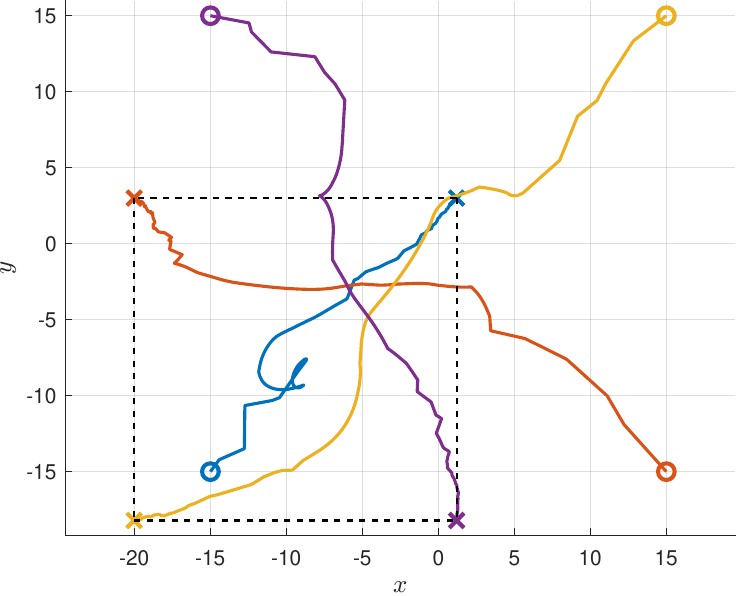}
        \caption{Local minimum is resolved due to asymmetries in the network topologies.}
        \label{fig:agents_traj_lm_solved}
    \end{subfigure}
    \caption{Trajectories of four agents in space seeking a given square-shaped formation while avoiding collisions.}
\end{figure}
\section{Conclusion} \label{sec:conclusion}
In this paper, we have introduced a consensus-based control strategy, tailored for multi-agent systems with single-integrator dynamics. Central to our approach is the use of an OtA broadcast protocol. This protocol exploits the superposition property of the wireless channel, leading to a drastic reduction of transmissions compared to approaches that avoid interference by using multiplexing.

We have advanced the results of~\cite{paper:molinari} by integrating a collision avoidance mechanism using artificial potential fields. This enhancement not only guarantees the safety of the agents but also demonstrates the adaptability and practicality of our control strategy in real-world scenarios.

A key finding of our research is the superior efficiency of our proposed controller compared to existing state-of-the-art methods. Notably, this efficiency becomes more pronounced as the complexity of the network increases in terms of the number of agents involved. This scalability is a crucial advantage, particularly in applications involving large-scale systems~\cite{paper:molinari}.

However, it is important to acknowledge that while our controller assures convergence, in certain ``pathological'' cases, the system may settle into a local minimum where agents cease movement without achieving the desired formation. Addressing this limitation by identifying and avoiding such cases forms a pivotal part of our future work. We also aim to extend our results to a broader class of systems, including agents with more complex dynamics, and to undertake experimental validation of our proposed controller using a group of real-world mobile robots.

\bibliographystyle{ieeetr}
\bibliography{lib}

\begin{thebibliography}{10}

\bibitem{paper:ren}
W.~Ren, R.~W. Beard, and E.~M. Atkins, ``Information consensus in multivehicle
  cooperative control,'' {\em IEEE Control systems magazine}, vol.~27, no.~2,
  pp.~71--82, 2007.

\bibitem{paper:gulzar}
M.~M. Gulzar, S.~T.~H. Rizvi, M.~Y. Javed, U.~Munir, and H.~Asif, ``Multi-agent
  cooperative control consensus: A comparative review,'' {\em Electronics},
  vol.~7, no.~2, p.~22, 2018.

\bibitem{paper:molinari}
F.~Molinari and J.~Raisch, ``Efficient consensus-based formation control with
  discrete-time broadcast updates,'' {\em 2019 IEEE 58th Conference on Decision
  and Control (CDC)}, pp.~4172--4177, 2019.

\bibitem{paper:falconi}
R.~Falconi, L.~Sabattini, C.~Secchi, C.~Fantuzzi, and C.~Melchiorri,
  ``Edge-weighted consensus-based formation control strategy with collision
  avoidance,'' {\em Robotica}, vol.~33, no.~2, pp.~332--347, 2015.

\bibitem{paper:sabattini}
L.~Sabattini, C.~Secchi, and C.~Fantuzzi, ``Potential based control strategy
  for arbitrary shape formations of mobile robots,'' in {\em 2009 IEEE/RSJ
  International Conference on Intelligent Robots and Systems}, pp.~3762--3767,
  2009.

\bibitem{paper:toyota}
R.~Toyota and T.~Namerikawa, ``Formation control of multi-agent system
  considering obstacle avoidance,'' in {\em 2017 56th Annual Conference of the
  Society of Instrument and Control Engineers of Japan (SICE)}, pp.~446--451,
  2017.

\bibitem{paper:yi}
X.~Yi, J.~Wei, D.~V. Dimarogonas, and K.~H. Johansson, ``Formation control for
  multi-agent systems with connectivity preservation and event-triggered
  controllers,'' {\em CoRR}, vol.~abs/1611.03105, 2016.

\bibitem{paper:almeida}
J.~Almeida, C.~Silvestre, A.~M. Pascoal, and P.~J. Antsaklis, ``Continuous-time
  consensus with discrete-time communication,'' in {\em 2009 European Control
  Conference (ECC)}, pp.~749--754, 2009.

\bibitem{paper:wang}
Z.~Wang, Y.~Zhao, Y.~Zhou, Y.~Shi, C.~Jiang, and K.~B. Letaief, ``Over-the-air
  computation: Foundations, technologies, and applications,'' {\em arXiv
  preprint arXiv:2210.10524}, 2022.

\bibitem{paper:zheng}
M.~Zheng, M.~Goldenbaum, S.~Stańczak, and H.~Yu, ``Fast average consensus in
  clustered wireless sensor networks by superposition gossiping,'' in {\em 2012
  IEEE Wireless Communications and Networking Conference (WCNC)},
  pp.~1982--1987, 2012.

\bibitem{paper:molinari2}
F.~Molinari, S.~Stanczak, and J.~Raisch, ``Exploiting the superposition
  property of wireless communication for average consensus problems in
  multi-agent systems,'' in {\em 2018 European Control Conference (ECC)},
  pp.~1766--1772, 2018.

\bibitem{molinari2021max}
F.~Molinari, N.~Agrawal, S.~Sta{\'n}czak, and J.~Raisch, ``Max-consensus over
  fading wireless channels,'' {\em IEEE Transactions on Control of Network
  Systems}, vol.~8, no.~2, pp.~791--802, 2021.

\bibitem{book:horn}
R.~A. Horn and C.~R. Johnson, {\em Matrix Analysis}.
\newblock Cambridge University Press, 1985.

\bibitem{paper:schwarz}
S.~Schwarz, ``New kinds of theorems on non-negative matrices,'' {\em
  Czechoslovak Mathematical Journal}, vol.~16, no.~2, pp.~285--295, 1966.

\bibitem{book:minc}
H.~Minc, {\em Nonnegative matrices}, vol.~170.
\newblock Wiley New York, 1988.

\bibitem{paper:wolfowitz}
J.~Wolfowitz, ``Products of indecomposable, aperiodic, stochastic matrices,''
  {\em Proceedings of the American Mathematical Society}, vol.~14, no.~5,
  pp.~733--737, 1963.

\end{thebibliography}
\end{document}